\documentclass[a4paper]{article}
\usepackage[utf8]{inputenc}
\usepackage[margin=1in]{geometry}
\usepackage{times}
\usepackage{dsfont}
\usepackage[switch]{lineno}
\usepackage[round]{natbib}
\usepackage{algorithm}
\usepackage{algpseudocode}
\usepackage[fleqn]{amsmath}
\usepackage{amssymb,amsthm}
\usepackage{tikz}
\usetikzlibrary{plotmarks,decorations.pathreplacing}
\usepackage{pgfplots}
\usepackage{subcaption}
\usepackage{xpatch}
\usepackage{url}

\makeatletter
\xpatchcmd{\algorithmic}{\ALG@tlm\z@}{\ALG@tlm\z@\leftmargin 0pt}{}{}
\makeatother
\algrenewcommand\algorithmicindent{1em}

\definecolor{jkblue}{rgb}{0,0.2,0.7}

\newlength\figureheight
\newlength\figurewidth

\newcommand{\var}{\operatorname{var}}
\newcommand{\RESS}{\text{RESS}}
\newcommand{\CAR}{\text{CAR}}
\newcommand{\E}{\mathds{E}}
\newcommand{\F}{\mathcal{F}}
\newcommand{\n}{{(n)}}
\newcommand{\m}{{(m)}}

\newtheorem{lemma}{Lemma}
\newtheorem{theorem}{Theorem}

\title{Probabilistic programming for birth-death models of evolution using an alive particle filter with delayed sampling}
\author{%
Jan Kudlicka\\
Uppsala University\\
Uppsala, Sweden\\
\and
Lawrence M. Murray\\
Uber AI Labs\\
San Francisco, CA, USA\\
\and
Fredrik Ronquist\\
Swedish Museum\\ of Natural History\\
Stockholm, Sweden\\
\and
Thomas B.~Sch\"on\\
Uppsala University\\
Uppsala, Sweden\\
}
\date{}

\begin{document}

\maketitle

{\small\textbf{Please cite this version:}

J.~Kudlicka, L.~M. Murray, F.~Ronquist, and T.~B. Sch\"on.
\newblock Probabilistic programming for birth-death models of evolution using
  an alive particle filter with delayed sampling.
\newblock In \emph{Conference on Uncertainty in Artificial Intelligence}, 2019.

\begin{center}
\begin{minipage}{.9\linewidth}
\begin{verbatim}
@inproceedings{kudlicka2019probabilistic,
  title={Probabilistic programming for birth-death models of evolution using
         an alive particle filter with delayed sampling},
  author={Kudlicka, Jan and Murray, Lawrence M. and Ronquist, Fredrik and
          Sch\"on, Thomas B.},
  booktitle={Conference on Uncertainty in Artificial Intelligence},
  year={2019}
}
\end{verbatim}
\end{minipage}
\end{center}}

\begin{abstract}
We consider probabilistic programming for birth-death models of evolution and introduce a new widely-applicable inference method that combines an extension of the alive particle filter (APF) with automatic Rao-Blackwellization via delayed sampling. Birth-death models of evolution are an important family of phylogenetic models of the diversification processes that lead to evolutionary trees. Probabilistic programming languages (PPLs) give phylogeneticists a new and exciting tool: their models can be implemented as probabilistic programs with just a basic knowledge of programming. The general inference methods in PPLs reduce the need for external experts, allow quick prototyping and testing, and accelerate the development and deployment of new models. We show how these birth-death models can be implemented as simple programs in existing PPLs, and demonstrate the usefulness of the proposed inference method for such models. For the popular BiSSE model the method yields an increase of the effective sample size and the conditional acceptance rate by a factor of 30 in comparison with a standard bootstrap particle filter. Although concentrating on phylogenetics, the extended APF is a general inference method that shows its strength in situations where particles are often assigned zero weight. In the case when the weights are always positive, the extra cost of using the APF rather than the bootstrap particle filter is negligible, making our method a suitable drop-in replacement for the bootstrap particle filter in probabilistic programming inference.
\end{abstract}

\section{Introduction}

The development of new probabilistic models of evolution is an important part of statistical phylogenetics. These models require inference algorithms that are able to cope with increased model complexity as well as the larger amount of observational data available today. Experts from several fields typically need to be involved, both to design bespoke inference algorithms, and to implement the new models and the inference algorithms in existing software or to develop new software from scratch. Probabilistic programming languages (PPLs) \citep[e.g.,][]{goodman2008church,tolpin2016design,mansinghka2014venture,paige2014compilation} have the potential to accelerate this: generative models are specified as simple programs and compiled into executable applications that include general inference engines. Writing models in PPLs requires just basic programming skills, and thus allows quick prototyping and testing.

Quite a few software applications for statistical phylogenetics exist today, including the popular MrBayes \citep{huelsenbeck2001mrbayes} and BEAST \citep{drummond2007beast}. They typically take a Bayesian approach and implement Markov chain Monte Carlo inference \citep[see review by][]{nascimento2017biologist}. Most of these applications do not allow the user to specify models outside of a predefined model space, which can be quite narrow. Even when adding new models is possible, it is usually a challenging task requiring not only good programming skills but also detailed knowledge of the software design and implementation. 

Statistical phylogeneticists recognize the benefits of software that supports the addition of new models and inference methods. For example, the design of BEAST 2 \citep{bouckaert2014beast} allows users to create and use custom modules. RevBayes \citep{hohna2016revbayes} goes even further: it uses a domain-specific probabilistic programming language for phylogenetics based on probabilistic graphical models \citep[e.g.,][]{koller2009probabilistic}. However, the language is not Turing-complete, which means it has some limitations. For example it does not allow unbounded recursion.

In this paper we concentrate on birth-death models of evolution, an important family of phylogenetic models. In these models, births correspond to lineage splits (speciation events) and deaths to extinction events. These models specify probability distributions of evolutionary trees and the task is to infer model parameters given a part of a complete tree that represents evolution of currently living species.

We take a step toward using PPLs in statistical phylogenetics: our main contribution is a new general inference algorithm based on an extension of the alive particle filter (APF) \citep{moral2015alive} combined with automatic Rao-Blackwellization via delayed sampling \citep{murray2018delayed}. We also show how to implement birth-death models of evolution in existing PPLs, and show the usefulness of our inference algorithm for such models. Interestingly, by using this algorithm we avoid sampling of birth and death rates. We believe that the algorithm may be of interest for other models with highly-informative observations. Finally, we prove that the estimator of the marginal likelihood in the extended APF is unbiased.

The rest of the paper is organized as follows: in Section~\ref{sec:background} we give a brief recapitulation of basic concepts in evolution and introduce probabilistic programming in more detail. We derive our inference algorithm and show how phylogenetic birth-death algorithms can be implemented in PPLs in Section~\ref{sec:methods}. We give implementations of two well-known phylogenetic birth-death models and compare several general inference algorithms for these models in Section~\ref{sec:experiments}. We offer some conclusions and ideas for future research in Section~\ref{sec:conclusion}.

\section{Background}
\label{sec:background}

\subsection{Speciation, extinction and phylogenies}
\label{subsec:evolution}

There are two types of events that play a significant role in the evolution of any species.
\begin{itemize}
    \item \emph{Speciation} occurs when the population of one species splits and eventually forms two new species.
    \item \emph{Extinction} occurs when the whole population of one species dies out. Species that are not extinct, i.e., species with individuals alive at the present time, are called \emph{extant}.
\end{itemize}

In phylogenetics, the \emph{before present} (BP) time is usually used for dating, i.e., if an event happened at time $\tau$ it means it happened $\tau$ time units ago. 

The result of an evolutionary process is a binary tree called the \emph{complete phylogeny}. A very simple example of a complete phylogeny is depicted in Figure~\ref{fig:complete_reconstructed}a. The nodes represent events and species at significant times:
\begin{itemize}
    \item the root node represents the most recent common ancestor (MRCA) of all species of interest,
    \item an internal node represents a speciation event,
    \item a leaf at $\tau=0$ (i.e. the present time) represents an extant species,
    \item a leaf at $\tau>0$ (i.e. in the past) represents an extinction event.
\end{itemize}
The length of edges---or \emph{branches} as they are called in phylogenetics---is the difference between the time of the parent and the child node.

The \emph{reconstructed} phylogeny is obtained from a complete tree by removing all subtrees that involve only extinct species. We will refer to this as pruning. An example of a reconstructed tree is depicted in Figure~\ref{fig:complete_reconstructed}b.

The reconstructed phylogeny represents the evolution of the extant species and only contains information that can be observed directly (the extant species) or \emph{reconstructed} by statistical analysis of the DNA sequences of extant species (the topology of the tree and the times of the speciation events).

\begin{figure}
\centering
\tikzstyle{observed}=[ultra thick]
\tikzstyle{unobserved}=[]
\begin{tikzpicture}[scale=1.4]
% Complete tree
\node at (-3,3.25) {(a) Complete tree};
\draw[observed] (-3, 1.75) -- (-3, 1.7);
\draw[observed] (-3.5, 1.7) -- (-2.5, 1.7);
\draw[observed] (-3.5, 1.7) -- (-3.5, 1.1);
\draw[observed] (-3.75, 1.1) -- (-3.25, 1.1);
\draw[observed] (-3.75, 1.1) -- (-3.75, 0);
\draw[observed] (-3.25, 1.1) -- (-3.25, 0);
\draw[observed] (-2.5, 1.7) -- (-2.5, 0.6);
\draw[observed] (-2.75, 0.6) -- (-2.25, 0.6);
\draw[observed] (-2.75, 0.6) -- (-2.75, 0);
\draw[observed] (-2.25, 0.6) -- (-2.25, 0);
\draw[unobserved] (-2.5, 1.4) -- (-1.5, 1.4) -- (-1.5, 1);
\draw[unobserved] (-1.75, 1) -- (-1.25, 1);
\draw[unobserved] (-1.75, 1) -- (-1.75, 0.7);
\draw[unobserved] (-2, 0.7) -- (-1.5, 0.7);
\draw[unobserved] (-2, 0.7) -- (-2, 0.4);
\draw[unobserved] (-1.5, 0.7) -- (-1.5, 0.2);
\draw[unobserved] (-1.25, 1) -- (-1.25, 0.6);
\draw[unobserved] (-3.75, 0.8) -- (-4.25, 0.8) -- (-4.25, 0.5);
\draw[unobserved] (-4.5, 0.5) -- (-4, 0.5);
\draw[unobserved] (-4.5, 0.5) -- (-4.5, 0.3);
\draw[unobserved] (-4, 0.5) -- (-4, 0.2);

\draw[->] (-0.4,1.2) -- (1.4, 1.2) node[midway,above]{Pruning};
\draw[->] (1.4,1.0) -- (-0.4, 1.0) node[midway,below]{Augmentation};

% Reconstructed tree
\node at (3,3.25) {(b) Reconstructed tree};
\draw[observed] (3, 1.75) -- (3, 1.7);
\draw[observed] (2.5, 1.7) -- (3.5, 1.7);
\draw[observed] (2.5, 1.7) -- (2.5, 1.1);
\draw[observed] (2.25, 1.1) -- (2.75, 1.1);
\draw[observed] (2.25, 1.1) -- (2.25, 0);
\draw[observed] (2.75, 1.1) -- (2.75, 0.3);
\draw[observed] (2.75, 0.3) -- (2.75, 0);
\draw[observed] (3.5, 1.7) -- (3.5, 1.45);
\draw[observed] (3.5, 1.45) -- (3.5, 0.6);
\draw[observed] (3.25, 0.6) -- (3.75, 0.6);
\draw[observed] (3.25, 0.6) -- (3.25, 0);
\draw[observed] (3.75, 0.6) -- (3.75, 0);

% Time axis
\draw[<-] (-5,2.8) -- (-5,-0.1) node[pos=0.5,xshift=-30,rotate=90]{\small Time $\tau$} ;
\draw[dotted] (-4.9, 0) -- (5, 0);
\draw (-4.95,0) -- (-5.05,0) node[xshift=-10,rotate=90]{$0$};
\draw (-4.95,1.7) -- (-5.05,1.7) node[xshift=-10,rotate=90]{$\tau_\text{MRCA}$};

% Annotation of the complete tree

\node (se) at (-3, 2.8) [anchor=north,align=center] {\small Root\\ \small (MRCA)};
\draw[->,densely dotted] (se) to[out=-90,in=90] (-3, 1.8);

\node (se) at (-1.7, 2.8) [anchor=north,align=center] {\small Unobserved\\ \small(hidden)\\ \small speciation};
\draw[->,densely dotted] (se) to[out=-120,in=30] (-2.45, 1.45);

\node (use) at (-4.2, 2.8) [anchor=north,align=center] {\small Observed\\ \small speciation};
\draw[->,densely dotted] (use) to[out=-80,in=140] (-3.55, 1.15);

\node (ee) at (-4.5, -0.5) [anchor=north] {\small Extinction};
\draw[->,densely dotted] (ee) to[out=80,in=-120] (-4.05, 0.15);

\node (ex) at (-3, -0.5) [anchor=north] {\small Extant species};
\draw[->,densely dotted] (ex) to[out=130,in=-90] (-3.75, -0.05);
\draw[->,densely dotted] (ex) to[out=110,in=-90] (-3.25, -0.05);
\draw[->,densely dotted] (ex) to[out=70,in=-90] (-2.75, -0.05);
\draw[->,densely dotted] (ex) to[out=50,in=-90] (-2.25, -0.05);

\draw[densely dotted] (-1.625,0.73) ellipse (0.5 and 0.8);
\node (us) at (-0.6, -0.5) [anchor=north,align=center] {\small Hidden (unobserved) subtree};
\draw[->,densely dotted] (us) to[out=110,in=-70] (-1.5, -0.1);

\end{tikzpicture}
\caption{Complete (on the left) vs. reconstructed (on the right) tree. The reconstructed tree shows only the evolution of the extant species.}
\label{fig:complete_reconstructed}
\end{figure}
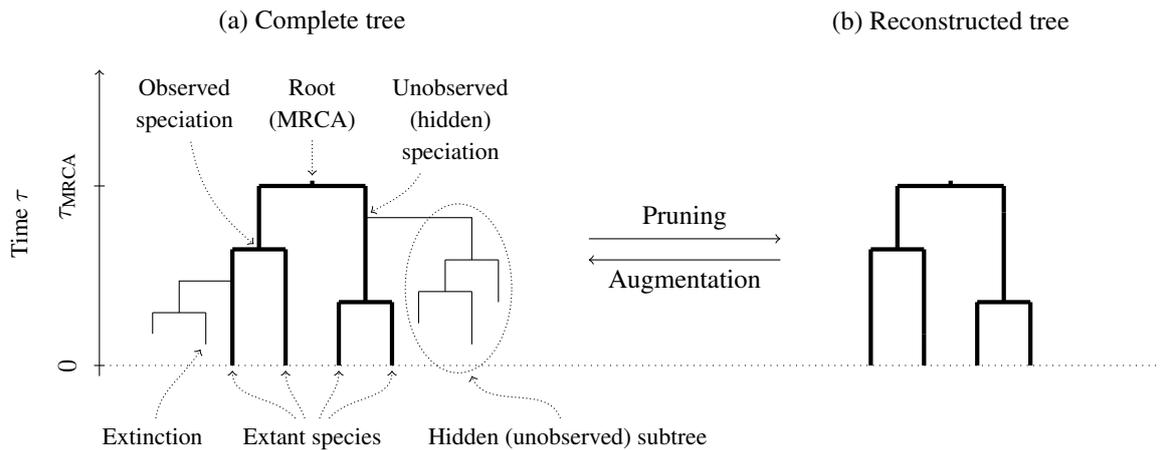

\subsection{Probabilistic programming}
\label{subsec:probprog}

The development of new probabilistic models and inference algorithms is a time-consuming and possibly error-prone process that usually requires skilled experts in probability, statistics and computer science. Probabilistic programming is a relatively new approach to solve this problem: generative models are expressed as computer programs in probabilistic programming languages (PPLs) with support for random variables and operations on them. Integral to PPLs are general inference engines that perform the inference in such programs. These engines estimate the distribution of all latent random variables conditioned on the observed data and use it to answer the queries of interest.

PPLs allow us to define and initialize random variables with a given distribution, for example:
\begin{quote}
\begin{algorithmic}
\State $x \sim \operatorname{Normal}(0, 1)$
\end{algorithmic}
\end{quote}\vspace{-0.5pc}
The program may use random variables as though any ordinary variable, and control the flow of the execution as shown in the following example:\begin{quote}
\begin{algorithmic}
\If {$x > 0.5$}
    \State $y \sim \operatorname{Normal}(x, 1)$
\Else
    \State $y \sim \operatorname{Exponential}(1)$
\EndIf
\end{algorithmic}
\end{quote}\vspace{-0.5pc}
Depending on the PPL, conditioning on the observed data might be specified explicitly or implicitly. The former means that conditioning on the observed data is a part of the program, for example:
\begin{quote}
\begin{algorithmic}
\State $x \sim \operatorname{Normal}(0, 1)$
\State $\mathbf{observe}\ 0.892 \sim \operatorname{Normal}(x, 1)$
\end{algorithmic}
\end{quote}\vspace{-0.5pc}
The latter implies that observed values of random variables are not part of the program, but instead specified at run time (e.g. as arguments).

The main general inference methods used in PPLs are adaptations of various inference algorithms for the universal setting described below, including Markov chain Monte Carlo (MCMC) \citep{metropolis1953equation, hastings1970monte}, sequential Monte Carlo (SMC) \citep{Gordon:1993,del2006sequential}, and Hamiltonian Monte Carlo (HMC) \citep{neal2011mcmc}.

There already exist quite a few PPLs today based on different programming paradigms, for example functional PPLs like Anglican \citep{tolpin2016design} and Venture \citep{mansinghka2014venture}; imperative Probabilistic C \citep{paige2014compilation}, Turing \citep{Ge2018}, Stan \citep{carpenter2017stan}, Edward \citep{Tran2016} and Pyro \citep{JMLR:v20:18-403}; and object-oriented Birch \citep{murray2018automated}.

\subsection{Programmatic model and SMC}

The execution of a probabilistic program can be modeled using a \emph{programmatic model} \citep{murray2018automated}.
Let $\{V_i\}_i$ denote a countable set of all random variables, both latent and observed, in a probabilistic program. This set might be infinite due to loops and recursion. During execution, whenever a random variable $V_i$ is encountered, its realization $v_i$ is drawn from a distribution associated with it.

Multiple executions of the program might in general encounter different subsets of the random variables (e.g., due to using random variables in conditional expressions) and encounter them in a different order (with the exception of the first one). For each random variable $V_j$ not encountered during the execution we set $v_j = \bot$ (the symbol $\bot$ represents an \emph{undefined} value). We will however assume that any execution encounters all \emph{observed} random variables and that these observations are encountered in the same order.

Let $\sigma$ denote a sequence of indices into $\{V_i\}$ specifying the order in which the random variables are encountered during an execution of the program, and let $|\sigma|$ denote the length of this sequence. Also, let $\{v_i\}_{i\in\sigma}$ denote the realizations of the random variables indexed by $\sigma$, i.e., $v_{\sigma[1]}, \dots, v_{\sigma[|\sigma|]}$. In a similar manner, we will use $\{V_i=v_i\}_{i\in\sigma}$ to denote $V_{\sigma[1]} = v_{\sigma[1]}, \dots, V_{\sigma[|\sigma|]} = v_{\sigma[|\sigma|]}$.

The index of the $k$-th encountered variable is given by a deterministic function $\operatorname{Ne}$ (for \textit{next}) of the realizations of the previously encountered random variables, so that
\begin{linenomath*}
\begin{equation*}
\sigma[k] = \operatorname{Ne}(\{v_i\}_{i\in\sigma[1:k-1]}),
\end{equation*}
\end{linenomath*}
where $\sigma[1{:}k{-}1]$ denotes the indices of the first $k-1$ encountered random variables. The function $\operatorname{Ne}$ is uniquely defined by the probabilistic program. If there are no more random variables to encounter, $\operatorname{Ne}$ returns $\bot$.

The $k$-th encountered random variable, $V_{\sigma[k]}$, is sampled from
\begin{linenomath*}
\begin{equation*}
V_{\sigma[k]} \sim p_{\sigma[k]}(\cdot|\operatorname{Pa}(\{v_i\}_{i\in\sigma[1:k-1]})),
\end{equation*}
\end{linenomath*} 
where $p_{\sigma[k]}$ is the distribution specified by the program, $\operatorname{Pa}$ (for \textit{parents}) is a deterministic function returning the parameters of this distribution, and again, it is a function of the realizations of the previously encountered random variables.

The joint distribution function encoded by the program can be given recursively (starting with $\sigma=[]$):
\begin{linenomath*}
\begin{equation*}
p(\{v_i\}_{i\not\in\sigma}|\sigma, \{V_j=v_j\}_{j\in\sigma}) = \left\{\begin{array}{l@{}r}
p(\{v_i\}_{i\not\in\sigma'}|\sigma', \{V_j=v_j\}_{j\in\sigma'}) \times p_\kappa(v_\kappa|\operatorname{Pa}(\{v_j\}_{j\in\sigma})) & \text{if\ } \kappa\ne\bot,\\[0.5pc]
1 & \hspace{-5pc}\text{if\ }\kappa=\bot \wedge \forall i\not\in\sigma: v_i=\bot,\\[0.5pc]
0 & \text{otherwise,}
\end{array}\right.
\end{equation*}
\end{linenomath*} 
where $\kappa = \operatorname{Ne}(\{v_i\}_{i\in\sigma})$ and $\sigma'$ is obtained from $\sigma$ by appending $\kappa$. The first case is the conditional probability chain rule, the remaining cases cover the situation where there are no more random variables to encounter.

We wish to sample from the posterior distribution $p(\{v_i\}_{i\not\in\gamma}|V_{\gamma[1]} = y_1, \dots, V_{\gamma[T]} = y_T)$, where $T$ denotes the number of observations, $y_t$ denotes the $t$-th observation and $\gamma$ denotes the sequence of indices of the observed random variables in $\{V_i\}$. The sequential nature of the joint distribution allows us to employ Sequential Monte Carlo methods \citep{del2006sequential} to sample from this posterior distribution, including the \emph{bootstrap particle filter} (BPF) summarized in Algorithm \ref{alg:bpf}. For the sake of brevity we have assumed that the last observation is also the last encountered random variable. In the pseudocode, $\operatorname{Cat}()$ denotes the categorical distribution with the given event probabilities. Variables denoted by $v$ are associative arrays (also known as maps or dictionaries) used to store the realizations of random variables ($v[i]$ denotes the realization of $V_i$). The \textsc{Propagate} function runs the program until it encounters an observation. 

\begin{algorithm}[t]
\caption{Bootstrap particle filter (BPF).}
\label{alg:bpf}
\begin{algorithmic}
\For {$n=1\textbf{\ to\ }N$} \Comment{{\small Initialize}}
    \State $v^\n_0 \gets \varnothing$; $w^\n_0 \gets 1$
\EndFor
\For {$t=1\textbf{\ to\ }T$}
    \For {$n=1\textbf{\ to\ }N$}
        
        \State $a \sim \operatorname{Cat}(\{w^\m_{t-1}/\sum_{l=1}^N w^{(l)}_{t-1}\}_{m=1}^N)$ \Comment{{\small Resample}}
        \State $v^\n_t \gets \Call{Propagate}{v_{t-1}^{(a)}}$ \Comment{{\small Propagate}}
        \State $w_t^\n \gets p_{\gamma[t]}(y_t|\operatorname{Pa}(v_t^\n))$ \Comment{{\small Weigh}}
        \State $v^\n_t[\gamma[t]] \gets y_t$
    \EndFor
\EndFor
%\For {$n=1\textbf{\ to\ }N$} \Comment{Run until the end}
%    \State $v^\n_{T+1} \gets \Call{Propagate}{v^\n_T}$
%\EndFor
\vspace*{1pc}
\Function {Propagate}{$v$} \Comment{{\small Run until next observe}}
    \State $\kappa \gets \operatorname{Ne}(v)$
    \While {$\kappa\not\in\gamma$}  % \wedge k\ne\bot
        \State $v[\kappa] \sim p_\kappa(\cdot|\operatorname{Pa}(v))$
        \State $\kappa \gets \operatorname{Ne}(v)$
    \EndWhile
    \State \Return $v$
\EndFunction
\end{algorithmic}

\end{algorithm}

Samples from the joint distribution and the corresponding weights can be used to estimate the expected value of a test function $h$ of interest:
\begin{linenomath*}
\begin{equation*}
\widehat\E[h] = \frac{\sum_n w_T^\n h\left(v_T^\n\right)}{\sum_n w_T^\n},
\end{equation*}
\end{linenomath*}
as well as to estimate the marginal likelihood $p(y_{1:T})$:
\begin{linenomath*}
\begin{equation*}
\widehat Z = \prod_{t=1}^T \frac{1}{N} \sum_{n=1}^N w_t^\n.
\end{equation*}
\end{linenomath*}

\section{Methods}
\label{sec:methods}

\subsection{Extended alive particle filter}

In the bootstrap particle filter, each particle is propagated by simulating the prior, and may make random choices that lead to a state with zero weight. In phylogenetic birth-death models this happens quite often: when simulating the evolution of subtrees that must ultimately become extinct, if any species happen to survive to the present time, the particle is assigned zero weight. In extreme cases, all particles have zero weight, and the BPF degenerates.

\cite{moral2015alive} considered this problem in a setting with indicator potentials (such as in approximate Bayesian computation), i.e. all weights being either zero or one. They proposed a modification of the BPF, where the resampling and propagation steps are repeated for particles that have weight zero until all particles have weight one. Details of the resulting alive particle filter (APF) as well as proofs of some of its theoretical properties can be found in \citet{moral2015alive}.

Although the original APF was designed specifically for indicator potentials, we have \textit{extended the algorithm to work with importance weights}, see Algorithm \ref{alg:apf} (the \textsc{Propagate} function is the same as in Algorithm \ref{alg:bpf}). The APF requires $N+1$ particles rather than $N$ in order to estimate the marginal likelihood without bias.

\begin{algorithm}[t]
\caption{Alive particle filter (APF).}
\label{alg:apf}
\begin{algorithmic}
\For {$n=1\textbf{\ to\ }N$} \Comment{{\small Initialize}}
    \State $v^\n_0 \gets \varnothing$; $w^\n_0 \gets 1$
\EndFor
\For {$t=1\textbf{\ to\ }T$}
    \State {\color{jkblue}$P_t \gets 0$}
    \For {$n=1\textbf{\ to\ }\color{jkblue}N+1$}
        \color{jkblue}\Repeat\color{black} \Comment{{\small Resample}}
            \State $a \sim \operatorname{Cat}(\{w^\m_{t-1}/\sum_{l=1}^N w^{(l)}_{t-1}\}_{m=1}^N)$  
            \State $v^\n_t \gets \Call{Propagate}{v_{t-1}^{(a)}}$ \Comment{{\small Propagate}}
            \State {\color{jkblue}$P_t \gets P_t + 1$}
            \State $w_t^\n \gets p_{\gamma[t]}(y_t|\operatorname{Pa}(v_t^\n))$ \Comment{{\small Weigh}}
        \color{jkblue}\Until{$w_t^\n > 0$}\color{black}
        \State $v^\n_t[\gamma[t]] \gets y_t$
    \EndFor
\EndFor
%\For {$n=1\textbf{\ to\ }N$} \Comment{Run until the end}
%    \State $v^\n_{T+1} \gets \Call{Propagate}{v^\n_T}$
%\EndFor
\end{algorithmic}
\end{algorithm}

At the $t$-th observe statement, if the weight of a particle is zero, the resampling and propagation steps are repeated. This procedure is repeated until the weights of all $N+1$ particles are positive. The APF counts the total number of propagations $P_t$ for each observation. The algorithm never uses the states or weights of the $N+1$-th particle, but propagations made using this particle are included in $P_t$, and used to calculate the unbiased estimate of the marginal likelihood $p(y_{1:T})$:
\begin{linenomath*}
\begin{equation*}
\widehat Z = \prod_{t=1}^T \frac{\sum_{n=1}^N w_t^\n}{P_t - 1}. 
\end{equation*}
\end{linenomath*}
The proof of unbiasedness can be found in Appendix \ref{sec:proof} in the supplementary material.

Unbiasedness of the marginal likelihood estimate opens for the possibility to use the APF within particle Markov chain Monte Carlo methods.

\subsection{Birth-death models as probabilistic programs}

Phylogentic birth-death models constitute a family of models where speciation (birth) events and extinction (death) events occur along the branches of a phylogenetic tree. Typically, the waiting times between events are exponentially distributed. In general, the rates of these exponential distributions do not remain constant but rather change continuously, discontinuously, or both. Some models assume that these rates further depend on a state variable that itself evolves discontinuously along the tree; in some cases the value of this state variable is given for the extant species. 

The \emph{constant-rate birth death (CRBD)} model \citep{kendall1948generalized} is the simplest birth-death model, where the speciation rate $\lambda$ and extinction rate $\mu$ remain constant over time. Pseudocode for generating phylogenetic trees using the CRBD model can be found in Appendix \ref{sec:gen_crbd} in the supplementary material.

Phylogenetic trees are unordered (i.e.~there is no specific ordering of the children of each internal node) and usually include the labels for the extant species.
To derive the likelihood of a complete labeled phylogenetic tree $\mathcal{T}$, let us first assume that the tree is ordered and unlabeled.
Let $\mathcal{T}_r$ denote the subtree rooted at the node $r$, and $\operatorname{Ch}(r)$ denote the children of this node. The likelihood of the subtree $\mathcal{T}_r$ can be expressed recursively (we have dropped conditioning on $\lambda$ and $\mu$ in the notation for brevity):
\begin{linenomath*}
\begin{equation*}
p(\mathcal{T}_r) = \left\{ \begin{array}{@{\hspace{0.1pc}}l@{}r}
\prod\limits_{c\in \operatorname{Ch}(r)}\hspace{-0.5pc}p(\mathcal{T}_c) & \text{if $r$ is the root node},\\[1pc]
\lambda e^{-(\lambda+\mu) \Delta_r} \hspace{-0.5pc}\prod\limits_{c\in \operatorname{Ch}(r)}\hspace{-0.5pc}p(\mathcal{T}_c) & \text{if $r$ is a speciation},\\[1pc]
\mu e^{-(\lambda+\mu) \Delta_r} & \text{if $r$ is an extinction},\\[1pc]
e^{-(\lambda+\mu) \Delta_r} & \text{if $r$ is an extant species},
\end{array}\right.
\end{equation*}
\end{linenomath*}
where $\Delta_r$ is the length of the branch between the node $r$ and its parent. 
%The factor 2 occurs in the formulas because the orientation of the two child subtrees does not matter.
If $r$ is a speciation event, no extinction occurs along the branch (factor $e^{-\mu \Delta_r}$) and the speciation happens after a waiting time $\Delta_r$ (factor $\lambda \exp^{-\lambda \Delta_r}$). If $r$ is an extinction event, no speciation occurs along the branch (factor $e^{-\lambda \Delta_r}$) and the extinction occurs after waiting time $\Delta_r$ (factor $\mu e^{-\mu \Delta_r}$). Finally, if $r$ is an extant species, neither extinction nor speciation occurs along the branch.

The likelihood of the complete, unordered and labeled phylogeny $\mathcal{T}$ is given by
\begin{linenomath*}
\begin{equation*}
p(\mathcal T) = \frac{2^{S+1}}{C!} p(\mathcal{T}_\text{root}),
\end{equation*}
\end{linenomath*}
where $S$ is the number of speciation events (excluding the root) and $C$ is the number of extant species. The factor $2^{S+1}$ represents the number of all possible orderings of the tree and there are $C!$ permutations of the labels of the extant species.
The likelihood can  be conveniently written as
\begin{linenomath*}
\begin{equation*}
p(\mathcal{T}) = \frac{2^{S+1}}{C!} \lambda^S \mu^X e^{-(\lambda+\mu)L},
\end{equation*}
\end{linenomath*}
where we have introduced $L$ to denote the sum of all branch lengths and $X$ to denote the number of extinction events.
The likelihood in other birth-death models that admit varying rates and/or include the state can be derived in a similar way.

The task of interest is to infer model parameters given a reconstructed tree. Recall that this tree is a part of the complete tree corresponding to the extant species and their ancestors. A naive approach to inference is to simulate trees from the generative model, prune back the extinct subtrees, and reject those for which the pruned tree does not equal the observed tree. Such an approach always results in rejection.

Instead, we turn the problem upside down: starting with the observed tree and \emph{augmenting} it with unobserved information to obtain a complete tree. Recalling Figure~\ref{fig:complete_reconstructed}, the observed tree is traversed in depth-first order. Along each branch, the generative model is used to simulate:
\begin{itemize}
    \item changes to the state (in models with state),
    \item changes to the speciation and extinction rates,
    \item hidden speciation events.
\end{itemize}
For each of the hidden speciation events, the model simulates the evolution of the new species (i.e.~a hidden subtree). If any portion of a hidden subtree survives to the present time, the weight is set to zero. If not, the current weight is doubled, since there are two possible orderings of the children created by a hidden speciation event on an observed branch.

If the examined branch ends with a speciation event, the algorithm observes $0 \sim \operatorname{Exponential}(\lambda)$. Finally, as there were no extinction events along the processed branch, the algorithm observes $0 \sim \operatorname{Poisson}(\mu \Delta)$. In models with state, if the branch ends at $\tau=0$ (i.e. the present time) we also condition on the simulated state being equal to the observed state. We trigger resampling at the end of each branch.

Let us return to the CRBD model in light of the discussion above. The likelihood of a proposed complete tree $\mathcal{T}'$ that is compatible with the observation (i.e. without any extant species in the hidden subtrees) is given by
\begin{linenomath*}
\begin{equation*}
q(\mathcal{T}') = \lambda^{H'} e^{-\lambda L_\text{obs}} \times 2^{S'} \lambda^{S'} \mu^{X'} e^{-(\lambda+\mu) L'},
\end{equation*}
\end{linenomath*}
where $H'$ denotes the number of all simulated hidden speciation events along the observed tree, $L_\text{obs}$ the sum of the branch lengths in the observed tree, $S'$ the number of speciation events in all hidden subtrees, $X'$ the number of extinction events and finally $L'$ denotes the sum of the branch lengths in the hidden subtrees. The factor $\lambda^{H'} e^{-\lambda L_\text{obs}}$ is related to the hidden speciation events, and the rest is the combined likelihood of all hidden subtrees.

The weight of the proposal $\mathcal{T}'$ is given by
\begin{linenomath*}
\begin{equation*}
w(\mathcal{T}') = \frac{2^{S_\text{obs}+1}}{C!} \times 2^{H'} \times \lambda^{S_\text{obs}} \times e^{-\mu L_\text{obs}},
\end{equation*}
\end{linenomath*}
where $S_\text{obs}$ is the number of observed speciation events. The factor $2^{S_\text{obs}+1}/C!$, related to the number of possible orderings and the number of labeling permutations, is used as the initial weight of each proposal. The factor $2^{H'}$ corresponds to doubling the weight for each hidden subtree, the factor $\lambda^{S_\text{obs}}$ is due to observing $0 \sim \operatorname{Exponential}(\lambda)$ at all observed speciations, and finally the factor $e^{-\mu L_\text{obs}}$ is due to observing $0 \sim \operatorname{Poisson}(\mu \Delta)$ for all observed branches. 

Multiplying the likelihood $q(\mathcal{T'})$ and the weight $w(\mathcal{T'})$ of the proposal and summing the event numbers and the branch lengths we get
\begin{linenomath*}
\begin{equation*}
q(\mathcal{T'}) w(\mathcal{T}') = p(\mathcal{T'}).
\end{equation*}
\end{linenomath*}

The factor $2^{S_\text{obs}+1}/C!$ in $w(\mathcal{T}')$ is constant for all proposals and we will omit it from the weight in the algorithms and experiments.

\subsection{Delayed sampling of the rates}

In a Bayesian setting, the parameters are associated with a prior distribution. Using the gamma distribution (or the exponential distribution as its special case) as a prior for the rates of speciation, extinction and state change is mathematically convenient since the gamma distribution is a conjugate prior for both the Poisson and the exponential likelihood.
Instead of sampling these parameters from the prior distribution before running the particle filter (which we refer to as \emph{immediate sampling}), we can exploit the conjugacy, which allows us to marginalize out the parameters and sample them after running the particle filter. Exploiting the conjugacy in a probabilistic program can be automated by an algorithm known as \emph{delayed sampling}~\citep{murray2018delayed}.

Consider the following prior:
\begin{linenomath*}
\begin{equation*}
\nu \sim \operatorname{Gamma}(k, \theta),
\end{equation*}
\end{linenomath*}
with $k \in \mathbb{N}$. When the program needs to make a draw from a Poisson distribution 
\begin{linenomath*}
\begin{equation*}
n \sim \operatorname{Poisson}(\nu \Delta),
\end{equation*}
\end{linenomath*}
it can instead make a draw from the marginalized distribution:
\begin{linenomath*}
\begin{equation*}
n \sim \operatorname {NegativeBinomial}\left(k, \frac{1}{1 + \Delta \theta}\right),
\end{equation*}
\end{linenomath*}
where $\operatorname{NegativeBinomial}(k, p)$ is the negative binomial distribution counting the number of failures given the number of successes $k$ and the probability of success $p$ in each trial. The distribution for $\nu$ is then updated to the posterior distribution according to
\begin{linenomath*}
\begin{equation*}
\nu \sim \operatorname{Gamma}\left(k + n, \frac{\theta}{1 + \Delta \theta}\right).
\end{equation*}
\end{linenomath*}

Similarly for variables distributed according to the exponential distribution, instead of drawing
\begin{linenomath*}
\begin{equation*}
\Delta \sim \operatorname{Exponential}(\nu),
\end{equation*}
\end{linenomath*}
the program makes a draw from the marginalized distribution:
\begin{linenomath*}
\begin{equation*}
\Delta \sim \operatorname{Lomax}\left(\frac{1}{\theta}, k\right),
\end{equation*}
\end{linenomath*}
where the first parameter denotes the scale and the second the shape of the Lomax distribution, and the distribution for $\nu$ is then updated to
\begin{linenomath*}
\begin{equation*}
\nu \sim \operatorname{Gamma}\left(k + 1, \frac{\theta}{1 + \Delta \theta}\right).
\end{equation*}
\end{linenomath*}
Using this strategy there is actually \textit{no need to sample the rates at all}; all draws involving these rates are replaced by draws from the negative binomial and the Lomax distributions with a consequent update of the rate distribution. Details of these conjugacy relationships can be found in Appendix \ref{sec:distributions} in the supplementary material.

\section{Experiments}
\label{sec:experiments}

We implemented the inference algorithms described above in the probabilistic programming language Birch~\citep{murray2018automated} and added support for the conjugacy relationships described in the previous section, so that Birch can provide automated delayed sampling for these. We also implemented two phylogenetic birth-death models, described in Sections \ref{subsec:crbd} and \ref{subsec:bisse} below.

We ran the inference for these models using different combinations of the inference method, the sampling strategy (immediate or delayed) and different number of particles $N$. For each combination we executed the program $M$ times, collected the estimates $\{\widehat Z_m\}_{m=1}^M$ of the marginal likelihood and calculated the relative effective sample size (\RESS):
\begin{linenomath*}
\begin{equation*}
\RESS = \frac{1}{M} \frac{\left(\sum_{m=1}^M \widehat Z_m\right)^2}{\sum_{m=1}^M \widehat Z_m^2},
\end{equation*}
\end{linenomath*}
as well as the conditional acceptance ratio (CAR) (see \citealp{murray2013disturbance} for more detail):
\begin{linenomath*}
\begin{equation*}
\CAR = \frac{1}{M} \left(2 \sum_{i=1}^M c_i - 1 \right),
\end{equation*}
\end{linenomath*}
where $c_i$ is the sum of the $i$ smallest elements in $\{\widehat Z_m\}_m$. We also calculated the sample variance $\var\log\widehat Z$.

For the experiments with the APF we also compared the total number of propagations with the number of propagations in the BPF by calculating
\begin{linenomath*}
\begin{equation*}
\rho = \frac{\sum_{m=1}^M P_m}{MNT},
\end{equation*}
\end{linenomath*}
where $P_m$ is the number of all propagations made during the $m$-th execution %, $N$ is the number of particles
and $T$ is the number of branches in the observation. Note that $NT$ is the number of propagations in the BPF.

\subsection{Constant-rate birth death model}
\label{subsec:crbd}

\begin{algorithm}[h!]
\caption{CRBD model as a probabilistic program.}
\label{alg:crbd}
\begin{algorithmic}
\State \textbf{Input:} \begin{itemize} % we use \State instead of \Require to ensure correct indentation
    \item $\mathcal{T}$ -- a pre-ordered list of nodes in the observation
    \item $k_\lambda, \theta_\lambda$ -- the shape and scale of the prior Gamma distribution for $\lambda$ ($k_\lambda \in \mathbb{N}$)
    \item $k_\mu, \theta_\mu$ -- the shape and scale of the prior Gamma distribution for $\mu$ ($k_\mu \in \mathbb{N}$)
\end{itemize}
\vspace*{1pc}
\State $\lambda \sim \operatorname{Gamma}(k_\lambda, \theta_\lambda)$
\State $\mu \sim \operatorname{Gamma}(k_\mu, \theta_\mu)$
\ForAll {$r \in \mathcal{T}$}
    \If {$r$ is the root}
        \State \textbf{continue}
    \EndIf
    \State $c_\text{hs} \sim \operatorname{Poisson}(\lambda \Delta_r)$
    \For {$i \gets 1\ \textbf{to}\ c_\text{hs}$}
        \State $\tau \sim \operatorname{Uniform}(\tau_r, \tau_r + \Delta_r)$
        \If {\Call{BranchSurvives}{$\tau$}}
            \State set the weight to 0 and \textbf{return}
        \EndIf
        \State double the weight
    \EndFor
    \If {$r$ has children}
        \State $\textbf{observe}\ 0 \sim \operatorname{Exponential}(\lambda)$
    \EndIf
    \State $\textbf{observe}\ 0 \sim \operatorname{Poisson}(\mu \Delta_r)$
    %\State \textbf{yield} the weight
\EndFor
\vspace*{1pc}
\Function{BranchSurvives}{$\tau, \lambda, \mu$}
\State $\Delta \sim \operatorname{Exponential}(\mu)$
\If {$\Delta \geq \tau$}
    \State \Return{true}
\EndIf
\State $c_b \sim \operatorname{Poisson}\left(\lambda \Delta\right)$
\For {$i \gets 1\ \textbf{to}\ c_b$}
    \State $\tau' \sim \operatorname{Uniform}(\tau - \Delta, \tau)$
    \If {\Call{BranchSurvives}{$\tau', \lambda, \mu$}}
        \State \Return true
    \EndIf
\EndFor
\State \Return{false}
\EndFunction
\end{algorithmic}
\end{algorithm}

Pseudocode for the probabilistic program implementing the \emph{constant-rate birth death (CRBD)} model is listed as Algorithm \ref{alg:crbd}. To sample speciation events along a branch we first sample a number of events from a Poisson distribution and then sample the time of each event from a uniform distribution. The implementation in Birch can be found in the supplementary material.

We used the phylogeny of cetaceans \citep{steeman2009radiation} as the observation. This phylogeny (Figure~\ref{fig:cetaceans} in the supplementary material) represents the evolution of whales, dolphins and porpoises and contains 87 extant species. We used $\operatorname{Gamma}(1, 1)$ as the prior for both the speciation and extinction rate. The results of experiments comparing BPF and APF with immediate or delayed sampling for different number of particles $N$, and running $M=200$ executions for each combination, are summarized in Table~\ref{tab:res_crbd} and Figure~\ref{fig:logz_crbd}.

\begin{figure}
    \setlength\figureheight{3cm}
    \setlength\figurewidth{7cm}
    \centering
    \input{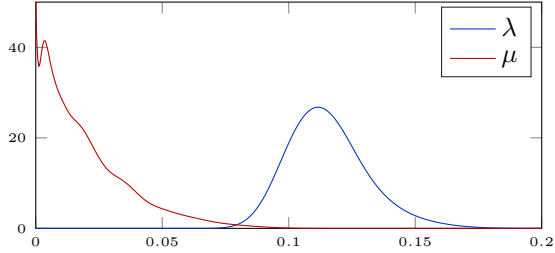}
    \caption{The posterior distributions for the speciation and extinction rates for the cetaceans using the CRBD model.}
    \label{fig:crbd_posterior}
\end{figure}

\begin{figure}[p]
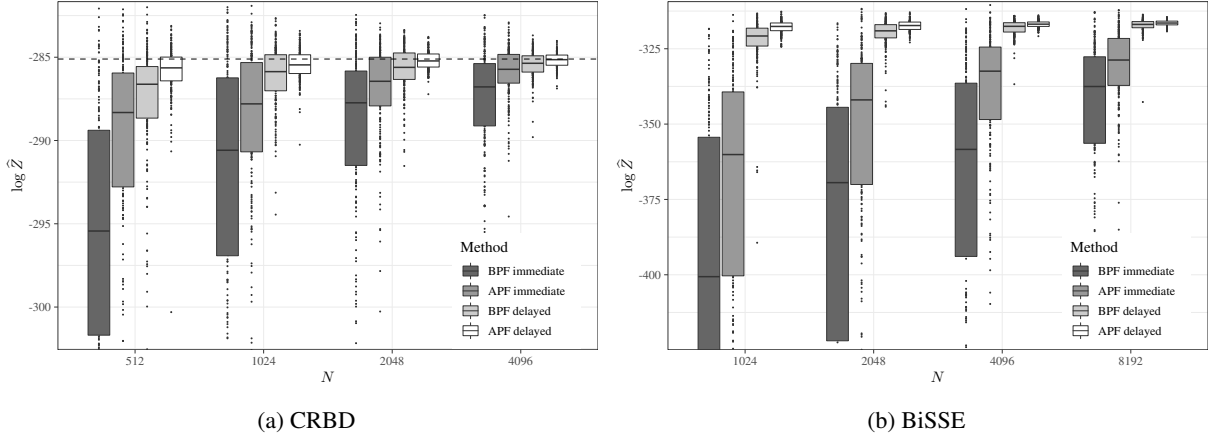

\begin{subfigure}{0.5\textwidth}
\centering
\resizebox{\textwidth}{!}{\input figs/logz_crbd}
\caption{CRBD}
\label{fig:logz_crbd}
\end{subfigure}
\begin{subfigure}{0.5\textwidth}
\centering
\resizebox{\textwidth}{!}{\input figs/logz_bisse}
\caption{BiSSE}
\label{fig:logz_bisse}
\end{subfigure}
\caption{Box plot of $\log \widehat Z$ for the CRBD (on the left) and BiSSE (on the right) model for different number of particles ($N$) and methods.  The lower and upper hinges correspond to the first and third quartile, whereas the midhinge corresponds to the median. Values outside of the interquartile range are shown as dots. The horizontal dashed line shows the true value of $\log Z$ for the CRBD model.}
\end{figure}

When using delayed sampling, the speciation and extinction rates are never sampled; the rates are instead represented by gamma distributions with parameters that are updated during the execution. Let $k_\lambda^\m, \theta_\lambda^\m, k_\mu^\m$ and $\theta_\mu^\m$ denote the final values of these parameters for a particle drawn from all particles in the $m$-th run with the probabilities proportional to their final weights. The posterior distributions for $\lambda$ and $\mu$ can be estimated by mixtures of gamma distributions:
\begin{linenomath*}
\begin{align*}
\lambda &\sim \frac{1}{\sum_{m=1}^M \widehat Z_m} \sum_{m=1}^M \widehat Z_m \operatorname{Gamma}\left(k_\lambda^\m, \theta_\lambda^\m\right),\\
\mu &\sim \frac{1}{\sum_{m=1}^M \widehat Z_m} \sum_{m=1}^M \widehat Z_m \operatorname{Gamma}\left(k_\mu^\m, \theta_\mu^\m\right).
\end{align*}
\end{linenomath*}

Figure~\ref{fig:crbd_posterior} depicts the posterior distributions for the speciation and extinction rates estimated using $M=1000$ runs of the APF with $N=4096$ particles.

\subsection{Binary-state speciation and extinction model}
\label{subsec:bisse}

The \emph{binary-state speciation and extinction (BiSSE)} model \citep{maddison2007estimating} introduces a binary state for species, denoted by $s \in \{0, 1\}$. Each state has its own (but constant) speciation and extinction rates, denoted by $\lambda_s$ and $\mu_s$. The waiting time between switching state is exponentially distributed with rates $q_{01}$ for switching from state 0 to state 1, and $q_{10}$ from state 1 to state 0. In our experiments we made a (common) assumption that $q_{01}=q_{10}=\varsigma$.

We used the same observation as \citet{rabosky2015model}, i.e., we extended the cetacean phylogeny with the state variable related to the body length of cetaceans obtained from \citet{slater2010diversity}. Data are available for 74 of the 87 extant species. The binary state 0 and 1 refers to the length being below and above the median. Again we used $\operatorname{Gamma}(1, 1)$ as the prior for $\lambda_0$, $\lambda_1$, $\mu_0$ and $\mu_1$, and $\operatorname{Gamma}(1, 10/820.28)$ as the prior for $\varsigma$ (the number in the denominator is the sum of all branch lengths). The initial state at the root is drawn from $\{0, 1\}$ with equal probabilities. The results for experiments comparing the BPF and the APF with immediate or delayed sampling for different number of particles $N$, and running $M=200$ executions for each combination, are summarized in Table~\ref{tab:res_bisse} and Figure~\ref{fig:logz_bisse}. When running the experiments using the BPF and immediate sampling, a certain fraction of the executions degenerated---from 25\% of the executions with 1024 particles down to 1.5\% of the executions with 4096 particles. These executions were excluded when calculating $\var\log\widehat Z$.

Our implementation of the BiSSE model can be found in the supplementary material.

\begin{table}[p]
\centering
\caption{Summary of the results of the experiments with the CRBD model using the cetacean phylogeny as the observation, priors $\lambda, \mu \sim \operatorname{Gamma}(1, 1)$, and $M=200$.}
\label{tab:res_crbd}
\begin{tabular}{@{}r|rrr|rrr|rrrr|rrrr@{}}
\multicolumn{1}{c}{} & \multicolumn{6}{c}{\bf Bootstrap particle filter (BPF)} & \multicolumn{8}{c}{\bf Alive particle filter (APF)} \\[0.5pc]
& \multicolumn{3}{c|}{Immediate sampling} & \multicolumn{3}{c|}{Delayed sampling} & \multicolumn{4}{c|}{Immediate sampling} & \multicolumn{4}{c}{Delayed sampling} \\[0.25pc]
$N$ & $\RESS$ & \hspace{-0.5pc}$\CAR$ & $\var$ & $\RESS$ & \hspace{-0.5pc}$\CAR$ & $\var$ & $\rho$ & \hspace*{-0.5pc}$\RESS$ & \hspace{-0.5pc}$\CAR$ & $\var$ & $\rho$ & \hspace*{-0.5pc}$\RESS$ & \hspace{-0.5pc}$\CAR$ & $\var$ \\
\hline
 512 &   0.02 &   0.04 &  334.2 &   0.13 &   0.23 &   31.6 &    1.8 &   0.11 &   0.15 &   50.7 &    1.7 &   0.40 &   0.46 &    2.7 \\
1024 &   0.11 &   0.12 &  117.5 &   0.28 &   0.35 &   12.9 &    1.8 &   0.14 &   0.18 &   20.2 &    1.7 &   0.54 &   0.55 &    0.8 \\
2048 &   0.14 &   0.17 &   52.2 &   0.47 &   0.47 &    8.7 &    1.7 &   0.29 &   0.32 &    7.6 &    1.7 &   0.73 &   0.69 &    0.3 \\
4096 &   0.18 &   0.23 &   17.2 &   0.67 &   0.63 &    0.7 &    1.7 &   0.36 &   0.42 &    2.7 &    1.7 &   0.84 &   0.76 &    0.2 \\
\end{tabular}
\end{table}

\begin{table}[p]
\centering
\caption{Summary of the results of the experiments with the BiSSE model using the cetacean phylogeny extended with information about the body length as the observation, priors $\lambda_0, \lambda_1, \mu_0, \mu_1 \sim \operatorname{Gamma}(1, 1)$ and $\varsigma \sim \operatorname{Gamma}(1, 10/820.28)$, and $M=200$.}
\label{tab:res_bisse}
\begin{tabular}{@{}r|rrr|rrr|rrrr|rrrr@{}}
\multicolumn{1}{c}{} & \multicolumn{6}{c}{\bf Bootstrap particle filter (BPF)} & \multicolumn{8}{c}{\bf Alive particle filter (APF)} \\[0.5pc]
& \multicolumn{3}{c|}{Immediate sampling} & \multicolumn{3}{c|}{Delayed sampling} & \multicolumn{4}{c|}{Immediate sampling} & \multicolumn{4}{c}{Delayed sampling} \\[0.25pc]
$N$ & $\RESS$ & \hspace{-.5pc}$\CAR$ & $\var$ & $\RESS$ & \hspace{-.5pc}$\CAR$ & $\var$ & $\rho$ & \hspace{-.5pc}$\RESS$ & \hspace{-.5pc}$\CAR$ & $\var$ & $\rho$ & \hspace{-.5pc}$\RESS$ & \hspace{-.5pc}$\CAR$ & $\var$ \\
\hline
1024 &   0.01 &   0.01 & 3382.2 &   0.06 &   0.09 &   72.4 &  10.0 &   0.01 &   0.01 & 2294.9 &    3.1 &   0.10 &   0.21 &    4.8 \\
2048 &   0.01 &   0.01 & 2954.0 &   0.09 &   0.15 &   22.2 &   6.6 &   0.02 &   0.02 & 1044.5 &    3.1 &   0.14 &   0.27 &    2.9 \\
4096 &   0.01 &   0.01 & 1894.1 &   0.22 &   0.27 &    7.6 &   5.9 &   0.01 &   0.01 &  614.3 &    3.1 &   0.34 &   0.43 &    1.3 \\
8192 &   0.02 &   0.02 &  968.4 &   0.28 &   0.35 &    6.1 &   3.9 &   0.02 &   0.03 &  160.9  &   3.0 &   0.54 &   0.55 &    0.8 \\
\end{tabular}
\end{table}

\section{Discussion and conclusion}
\label{sec:conclusion}

In this paper we introduced a new general inference method for probabilistic programming combining an extended alive particle filter (APF) with delayed sampling, and proved that the resulting estimate of the marginal likelihood is unbiased. We showed how phylogenetic birth-death models can be implemented in probabilistic programming languages, in particular, we considered two models---CRBD and BiSSE and their implementation in the probabilistic programming language Birch. We showed the strength of this inference method for these models compared to the standard bootstrap particle filter (BPF) (Tables~\ref{tab:res_crbd} and \ref{tab:res_bisse}, and Figures~\ref{fig:logz_crbd} and \ref{fig:logz_bisse}): for the BiSSE model using 8192 particles we increased $\RESS$ approximately 29 times, $\CAR$ approximately 30 times and lowered $\var\log\widehat Z$ more than  1150 times at the cost of running 3 times more propagations.

The extended APF is a suitable drop-in replacement for the BPF for black-box probabilistic programs. If a program produces only positive weights, the APF produces the same result as the BPF at the overhead of just one extra particle, used to estimate the marginal likelihood. On the other hand, if the program can produce zero weights, the APF behaves much more reasonably than the BPF: resampling and propagation are repeated until all particles have positive weight. This may seem equivalent to using the BPF with a higher number of particles ($\rho$ times more to be precise), but this is not the case: the number of propagations $P_t$ is not the same throughout the execution, but rather adapted dynamically for each $t$. This simplifies the tuning of the number of particles for such models.

The learning of rates in birth-death models sits in the context of broader problems in phylogenetics, such as the learning of trees. Interesting future work would be to consider whether models and methods for learning rates can be combined with models and methods for learning tree structures for joint inference.

\subsubsection*{Acknowledgements}

The authors wish to thank Johannes Borgstr\"om, Viktor Senderov and Andreas Lindholm for their useful feedback. This research was financially supported by the Swedish Foundation for Strategic Research (SSF) via the project ASSEMBLE and by the Swedish Research Council grants 2013-4853, 2014-05901 and 2017-03807.

\begingroup
\renewcommand{\section}[2]{}%
\subsubsection*{References}
\bibliography{kudlicka2019probabilistic}
\endgroup

\clearpage

\begin{center}
\vspace*{3pc}
\Large\uppercase{\bf Supplementary material}
\end{center}

\clearpage

\appendix

\section{Proof of the unbiasedness of the marginal likelihood estimator of the extended APF}
\label{sec:proof}

In this section we prove that the marginal likelihood estimator
\begin{linenomath*}
\begin{equation*}
\widehat Z = \prod_{t=1}^T \frac{\sum\limits_{n=1}^N w_t^\n}{P_t - 1},
\end{equation*}
\end{linenomath*}
produced by the extended alive particle filter (APF) for the state space model (Figure~\ref{fig:ssm})
\begin{linenomath*}
\begin{align*}
x_0 &\sim p(x_0), \\
x_t &\sim f_t(x_t|x_{t-1}), \text{\ for\ }t=1,2,\dots, T,\\
y_t &\sim g_t(y_t|x_t),
\end{align*}
\end{linenomath*}
is unbiased in the sense that $\E[\widehat Z] = p(y_{1:T})$.

The structure of our proof is similar to that of \cite{pitt2012some} for the Auxiliary Particle Filter. Let $\F_t = \{x_t^\n, w_t^\n\}_{n=1}^N$ denote the internal state of the particle filter, i.e., the states and weights of all particles, at time~$t$. 

\begin{lemma}
\label{lemma:1}
\begin{linenomath*}
\begin{equation*}
\E\left[\frac{\sum_{n=1}^N w_t^\n}{P_t - 1} \middle| \F_{t-1} \right] = \sum_{n=1}^N \frac{w_{t-1}^\n}{\sum_{m=1}^N w_{t-1}^\m} p\left(y_t \middle| x_{t-1}^\n\right).
\end{equation*}
\end{linenomath*}
\end{lemma}

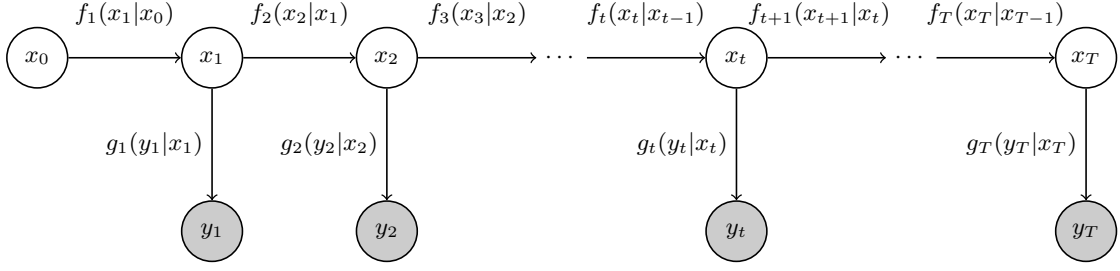
\begin{figure}
\centering
\tikzstyle{state}=[shape=circle,draw,inner sep=0pt,minimum size=8mm]
\tikzstyle{visible state}=[state,fill=black!20]
\begin{tikzpicture}[->,node distance=2.3cm,semithick]
\small
\node[state] (x0) {$x_0$};
\node[state] (x1) [right of=x0] {$x_1$};
\node[state] (x2) [right of=x1] {$x_2$};
\node (etc) [right of=x2] {$\cdots$};
\node[state] (xt) [right of=etc] {$x_t$};
\node (etc2) [right of=xt] {$\cdots$};
\node[state] (xT) [right of=etc2] {$x_T$};
\node[visible state] (y1) [below of=x1] {$y_1$};
\node[visible state] (y2) [below of=x2] {$y_2$};
\node[visible state] (yt) [below of=xt] {$y_t$};
\node[visible state] (yT) [below of=xT] {$y_T$};
\draw (x0) -- node[above,yshift=0.3cm]{$f_1(x_1|x_0)$} (x1);
\draw (x1) -- node[above,yshift=0.3cm]{$f_2(x_2|x_1)$} (x2);
\draw (x2) -- node[above,yshift=0.3cm]{$f_3(x_3|x_2)$} (etc);
\draw (etc) -- node[above,xshift=0cm,yshift=0.3cm]{$f_t(x_t|x_{t-1})$} (xt);
\draw (xt) -- node[above,xshift=-0.1cm,yshift=0.3cm]{$f_{t+1}(x_{t+1}|x_t)$} (etc2);
\draw (etc2) -- node[above,yshift=0.3cm]{$f_T(x_T|x_{T-1})$} (xT);
\draw (x1) -- node[left]{$g_1(y_1|x_1)$} (y1);
\draw (x2) -- node[left]{$g_2(y_2|x_2)$} (y2);
\draw (xt) -- node[left]{$g_t(y_t|x_t)$} (yt);
\draw (xT) -- node[left]{$g_T(y_T|x_T)$} (yT);
\end{tikzpicture}
\caption{Graphical model of the state space model.}
\label{fig:ssm}
\end{figure}

\begin{proof}
In the interest of brevity, we will omit conditioning on $\F_{t-1}$ in the notation. For each particle, the APF constructs a candidate sample $x'$ by drawing a sample from $\{x_{t-1}^\n\}$ with the probabilities proportional to the weights $\{w_{t-1}^\n\}$ and propagating it forward to time $t$ such that
\begin{linenomath*}
\begin{equation*}
x' \sim \sum_{n=1}^N \frac{w_{t-1}^\n}{\sum_{m=1}^N w_{t-1}^\m} f_t\left(x' \middle| x_{t-1}^\n\right). 
\end{equation*}
\end{linenomath*}
If $g_t(y_t|x') = 0$, the candidate sample is rejected and the procedure is repeated until acceptance (when $g_t(y_t|x') > 0$).

Let $A_t = \{x': g_t(y_t|x') > 0\}$.  The acceptance probability $p_{A_t}$ is then given by
\begin{linenomath*}
\begin{equation*} p_{A_t} = \int \mathbf{1}_{A_t}(x') \sum_{n=1}^N \frac{w_{t-1}^\n} {\sum_{m=1}^N w_{t-1}^\m} f_t\left(x' \middle| x_{t-1}^\n\right) dx',
\end{equation*}
\end{linenomath*}
where $\mathbf{1}$ denotes the indicator function.

The accepted samples are distributed according to the following distribution:
\begin{linenomath*}
\begin{equation*}
x_t \sim \frac{\mathbf{1}_{A_t}(x_t)}{p_{A_t}} \sum_{n=1}^N \frac{w_{t-1}^\n}{\sum_{m=1}^N w_{t-1}^\m} f_t\left(x_t \middle| x_{t-1}^\n\right).
\end{equation*}
\end{linenomath*}
The expected value of the weight $w_t = g_t(y_t|x_t)$ of an accepted sample is given by
\begin{linenomath*}
\begin{equation*}
\E[w_t] = \int g_t(y_t|x_t) \frac{\mathbf{1}_{A_t}(x_t)}{p_{A_t}} \sum_{n=1}^N \frac{w_{t-1}^\n}{\sum_{m=1}^N w_{t-1}^\m} f_t\left(x_t \middle| x_{t-1}^\n\right) dx_t.
\end{equation*}
\end{linenomath*}
The factor $\mathbf{1}_{A_t}(x_t)$ can be omitted since $\mathbf{1}_{A_t}(x_t) = 0 \Leftrightarrow g_t(y_t|x_t) = 0$, resulting in
\begin{linenomath*}
\begin{align*}
\E[w_t] &= \int \frac{1}{p_{A_t}} \sum_{n=1}^N \frac{w_{t-1}^\n}{\sum_{m=1}^N w_{t-1}^\m} f_t\left(x_t \middle| x_{t-1}^\n\right) g_t(y_t|x_t) dx_t \\ 
&= \frac{1}{p_{A_t}} \sum_{n=1}^N \frac{w_{t-1}^\n}{\sum_{m=1}^N w_{t-1}^\m} \int f_t\left(x_t \middle| x_{t-1}^\n\right) g_t(y_t|x_t) dx_t \\
&= \frac{1}{p_{A_t}} \sum_{n=1}^N \frac{w_{t-1}^\n}{\sum_{m=1}^N w_{t-1}^\m}\ p\left(y_t \middle| x_{t-1}^\n\right).
\end{align*}
\end{linenomath*}

The APF repeats drawing new samples until $N+1$ samples have been accepted. The total number of draws of candidate samples at time $t$, $P_t$, is itself a random variable distributed according to the negative binomial distribution
\begin{linenomath*}
\begin{equation*}
P(P_t = D) = \binom{D-1}{(N+1)-1} p_{A_t}^{N+1} (1-p_{A_t})^{D-(N+1)}
\end{equation*}
\end{linenomath*}
with the support $D \in \{N+1, N+2, N+3, \dots\}$.

Finally, using the fact that $\E[w_t]$ does not depend on the value of $P_t$,
\begin{linenomath*}
\begin{align*}
\E\left[\frac{\sum_{n=1}^N w_t^\n}{P_t - 1}\right] &= \sum_{D=N+1}^\infty \frac{N \E[w_t]}{D-1} P(P_t = D) = \sum_{D=N+1}^\infty \frac{N \E[w_t]}{D-1} \binom{D-1}{N} p_{A_t}^{N+1} (1-p_{A_t})^{D-(N+1)} \\
&= N \E[w_t] \sum_{D=N+1}^\infty \frac{1}{D-1} \binom{D-1}{N} p_{A_t}^{N+1} (1-p_{A_t})^{D-(N+1)} \\
&= N \E[w_t] \sum_{D=N+1}^\infty \frac{1}{D-1} \frac{(D-1) (D-2)!}{N (N-1)! (D-(N+1))!} p_{A_t}^{N+1} (1-p_{A_t})^{D-(N+1)} \\
&= \E[w_t] p_{A_t}^{N+1} \sum_{D=N+1}^\infty \binom{D-2}{D-(N+1)} (1-p_{A_t})^{D-(N+1)} \\
& \text{(using the binomial theorem)}\\
&= \E[w_t] p_{A_t}^{N+1} p_{A_t}^{-N} = \frac{1}{p_{A_t}} \sum_{n=1}^N \frac{w_{t-1}^\n}{\sum_{m=1}^N w_{t-1}^\m}\ p\left(y_t \middle| x_{t-1}^\n\right) p_{A_t} \\
&= \sum_{n=1}^N \frac{w_{t-1}^\n}{\sum_{m=1}^N w_{t-1}^\m}\ p\left(y_t \middle| x_{t-1}^\n\right).
\end{align*}
\end{linenomath*}
\end{proof}

\begin{lemma}
\label{lemma:2}
\begin{linenomath*}
\begin{equation*}
\E\left[\frac{\sum_{n=1}^N w_t^\n p\left(y_{t+1:t'} \middle| x_t^\n\right)}{P_t - 1} \middle| \F_{t-1} \right] = \sum_{n=1}^N \frac{w_{t-1}^\n}{\sum_{m=1}^N w_{t-1}^\m} p\left(y_{t:t'} \middle| x_{t-1}^\n\right).
\end{equation*}
\end{linenomath*}
\end{lemma}

\begin{proof}
Similar to the proof of Lemma \ref{lemma:1} we have that
\begin{linenomath*}
\begin{align*}
\E[w_t p(y_{t+1:t'}|x_t)] &= \int \frac{1}{p_{A_t}} \sum_{n=1}^N \frac{w_{t-1}^\n}{\sum_{m=1}^N w_{t-1}^\m} f_t\left(x_t \middle| x_{t-1}^\n\right) g_t(y_t|x_t) p(y_{t+1:t'}|x_t) dx_t \\ 
&= \frac{1}{p_{A_t}} \sum_{n=1}^N \frac{w_{t-1}^\n}{\sum_{m=1}^N w_{t-1}^\m} \int f_t\left(x_t \middle| x_{t-1}^\n\right) g_t(y_t|x_t) p(y_{t+1:t'}|x_t) dx_t \\
&= \frac{1}{p_{A_t}} \sum_{n=1}^N \frac{w_{t-1}^\n}{\sum_{m=1}^N w_{t-1}^\m}\ p\left(y_{t:t'} \middle| x_{t-1}^\n\right)
\end{align*}
\end{linenomath*}
and using this result we have that
\begin{linenomath*}
\begin{align*}
\E\left[\frac{\sum_{n=1}^N w_t^\n p\left(y_{t+1:t'} \middle| x_t^\n\right)}{P_t - 1}\right] &= \sum_{D=N+1}^\infty \frac{N \E[w_t p(y_{t+1:t'}|x_t)]}{D-1} \binom{D-1}{N} p_{A_t}^{N+1} (1-p_{A_t})^{D-(N+1)} \\
&= N \E[w_t p(y_{t+1:t'}|x_t)] \sum_{D=N+1}^\infty \frac{1}{D-1} \binom{D-1}{N} p_{A_t}^{N+1} (1-p_{A_t})^{D-(N+1)} \\
&= N \E[w_t p(y_{t+1:t'}|x_t)] \frac{p_{A_t}}{N} = N \frac{1}{p_{A_t}} \sum_{n=1}^N \frac{w_{t-1}^\n}{\sum_{m=1}^N w_{t-1}^\m}\ p\left(y_{t:t'} \middle| x_{t-1}^\n\right) \frac{p_{A_t}}{N} \\
&= \sum_{n=1}^N \frac{w_{t-1}^\n}{\sum_{m=1}^N w_{t-1}^\m}\ p\left(y_{t:t'} \middle| x_{t-1}^\n\right).
\end{align*}
\end{linenomath*}
\end{proof}

\begin{lemma}
\label{lemma:3}
\begin{linenomath*}
\begin{equation*}
\E\left[\prod_{t'=t-h}^t \frac{\sum_{n=1}^N w_{t'}^\n}{P_{t'} - 1} \middle| \F_{t-h-1} \right] = \sum_{n=1}^N \frac{w_{t-h-1}^\n}{\sum_{m=1}^N w_{t-h-1}^\m} p\left(y_{t-h:t} \middle| x_{t-h-1}^\n\right).
\end{equation*}
\end{linenomath*}
\end{lemma}

\begin{proof}
By induction.

The base step for $h=0$ was proved in Lemma \ref{lemma:1}.

In the induction step, let us assume that the equality holds for $h$ and prove it for $h+1$:
\begin{linenomath*}
\begin{align*}
\E\left[\prod_{t'=t-h-1}^t \frac{\sum_{n=1}^N w_{t'}^\n}{P_t' - 1} \middle| \F_{t-h-2} \right]
&= \E\left[\E\left[\prod_{t'=t-h}^t \frac{\sum_{n=1}^N w_{t'}^\n}{P_t' - 1} \middle| \F_{t-h-1} \right] \frac{\sum_{n=1}^N w_{t-h-1}^\n}{P_{t-h-1} - 1} \middle| \F_{t-h-2} \right]\\
& \text{(using the induction assumption)}\\
&= \E\left[   \sum_{n=1}^N \frac{w_{t-h-1}^\n}{\sum_{m=1}^N w_{t-h-1}^\m} p\left(y_{t-h:t} \middle| x_{t-h-1}^\n\right) \frac{\sum_{n=1}^N w_{t-h-1}^\n}{P_{t-h-1} - 1} \middle| \F_{t-h-2} \right] \\
&= \E\left[   \sum_{n=1}^N \frac{w_{t-h-1}^\n}{P_{t-h-1} - 1} p\left(y_{t-h:t} \middle| x_{t-h-1}^\n\right) \middle| \F_{t-h-2} \right] \\
& \text{(using Lemma \ref{lemma:2})}\\
&= \sum_{n=1}^N \frac{w_{t-h-2}^\n}{\sum_{m=1}^N w_{t-h-2}^\m} p\left(y_{t-h-1:t} \middle| x_{t-h-2}^\n\right).
\end{align*}
\end{linenomath*}
\end{proof}

\begin{theorem}
\begin{linenomath*}
\begin{equation*}
\E\left[\prod_{t=1}^T \frac{\sum\limits_{n=1}^N w_t^\n}{P_t - 1}\right] = p(y_{1:T}).
\end{equation*}
\end{linenomath*}
\end{theorem}

\begin{proof}
Using Lemma \ref{lemma:3} with $t=T, h=T-1$ and
\begin{linenomath*}
\begin{equation*}
\E\left[\frac{1}{N} \sum_{n=1}^N p\left(y_{1:T} \middle| x_0^\n\right)\right] = p(y_{1:T}).
\end{equation*}
\end{linenomath*}
\end{proof}

\section{Generative model for CRBD}
\label{sec:gen_crbd}

The pseudocode for generating phylogenetic trees using the CRBD model is listed in Algorithm \ref{alg:gen_crbd}.

\begin{algorithm}
\caption{Pseudocode for generating trees using the CRBD model.}
\label{alg:gen_crbd}
\begin{algorithmic}
\Function{crbd}{$\tau_\text{orig}$}
    \State \Return ($\tau_\text{orig}$, \{\Call{crbd'}{$\tau_\text{orig}$}\})
\EndFunction
\\
\Function{crbd'}{$\tau$}
\State $\Delta \sim \operatorname{Exponential}(\lambda + \mu)$
\State $\tau' \gets \tau - \Delta$
\If {$\tau' < 0$}
    \State \Return (0, $\varnothing$)
\EndIf
\State $e \sim \operatorname{Cat}\left(p_1=\frac{\lambda}{\lambda + \mu}, p_2=\frac{\mu}{\lambda + \mu}\right)$
\If {$e = 1$}
    \State \Return ($\tau'$, \{\Call{crbd'}{$\tau'$}, \Call{crbd'}{$\tau'$}\})
\Else
    \State \Return ($\tau'$, $\varnothing$)
\EndIf
\EndFunction
\end{algorithmic}

\end{algorithm}

\section{Relevant conjugacy relationships}
\label{sec:distributions}

\subsection{Negative binomial and Lomax distribution}

\subsubsection*{Negative binomial distribution}

Parameters: number of successes $k>0$ before the experiment is stopped, probability of success $p\in(0, 1)$

Probability mass function:
\begin{linenomath*}
\begin{align*}
f(r|k, p) &= \binom{r+k-1}{k-1} p^k (1-p)^r \text{\ for\ }r\in \mathbb {N} \cup \{0\},
\end{align*}
\end{linenomath*}
where $r$ is the number of failures.

\subsubsection*{Lomax distribution}

Parameters: scale $\lambda>0$, shape $\alpha>0$

Probability density function:
\begin{linenomath*}
\begin{equation*}
f(\Delta|\lambda, \alpha)= \frac{\alpha}{\lambda } \left(1+\frac{\Delta}{\lambda}\right)^{-(
\alpha+1)} \text{\ for\ } \Delta \geq 0
\end{equation*}
\end{linenomath*}

\subsection{Conjugacy relationships}

\subsubsection*{Gamma-Poisson mixture}

Prior distribution: $\nu \sim \operatorname{Gamma}(k, \theta)$ with the probability density function
\begin{linenomath*}
\begin{equation*}
f(\nu|k, \theta) = \frac{1}{\Gamma(k)\theta^{k}} \nu^{k-1} e^{-\nu/\theta} \text{\ for\ } \nu > 0
\end{equation*}
\end{linenomath*}

Likelihood: $n \sim \operatorname{Poisson}(\nu \Delta)$ with the probability mass function
\begin{linenomath*}
\begin{equation*}
f(n|\lambda) = \frac{\lambda^n}{n!} e^{-\lambda} \text{\ for\ } n\in \mathbb {N} \cup \{0\},
\end{equation*}
\end{linenomath*}
where $\lambda = \nu \Delta$.

Prior predictive distribution ($k \in \mathbb{N}$):
\begin{linenomath*}
\begin{align*}
f(n|k, \theta) &= \int_0^\infty \frac{1}{\Gamma(k)\theta^{k}} \nu^{k-1} e^{-\nu/\theta} \times \frac{(\nu \Delta)^n}{n!} e^{-\nu \Delta} d\nu = \frac{\Delta^n}{n! (k-1)! \theta^k} \int_0^\infty \nu^{n+k-1} e^{-\nu(1/\theta+\Delta)} d\nu \\
&= \frac{\Delta^n}{n! (k-1)! \theta^k} \left(\frac{1}{\theta} + \Delta\right)^{-(n+k)} (n+k-1)! = \binom{n+k-1}{k-1} \left(\frac{1}{1+\Delta\theta}\right)^k \left(1-\frac{1}{1+\Delta\theta}\right)^n \\
n|k, \theta &\sim \operatorname{NegativeBinomial}\left(k, \frac{1}{1+\Delta\theta}\right)
\end{align*}
\end{linenomath*}

Posterior distribution:
\begin{linenomath*}
\begin{align*}
f(\nu|n) &\propto \frac{1}{\Gamma(k)\theta^{k}} \nu^{k-1} e^{-\nu/\theta} \times \frac{(\nu \Delta)^n}{n!} e^{-\nu \Delta} \propto \nu^{k+n-1} e^{-\nu\left(1/\theta+\Delta\right)} = \nu^{(k+n)-1} e^{-\nu/\left(\frac{\theta}{1+\Delta\theta}\right)} \\
\nu|n &\sim \operatorname{Gamma}\left(k + n, \frac{\theta}{1 + \Delta \theta}\right)
\end{align*}
\end{linenomath*}

\subsubsection*{Gamma-exponential mixture}

Prior distribution: $\nu \sim \operatorname{Gamma}(k, \theta)$

Likelihood: $\Delta \sim \operatorname{Exponential}(\nu)$ with the probability density function
\begin{linenomath*}
\begin{equation*}
f(\Delta|\nu) = \nu e^{-\nu \Delta} \text{\ for\ } \Delta \geq 0
\end{equation*}
\end{linenomath*}

Prior predictive distribution:
\begin{linenomath*}
\begin{align*}
f(\Delta|k, \theta) &= \int_0^\infty \frac{1}{\Gamma(k)\theta^{k}} \nu^{k-1} e^{-\nu/\theta} \times \nu e^{-\nu \Delta} d\nu = \frac{1}{\Gamma(k)\theta^k} \int_0^\infty \nu^k e^{-\nu(1/\theta+\Delta)} d\nu \\
& = \frac{1}{\Gamma(k)\theta^k} \left(\frac{1}{\theta}+\Delta\right)^{-(k+1)} \Gamma(k+1) = \frac{k}{\theta^k} \left(\frac{1}{\theta}+\Delta\right)^{-(k+1)} = k\theta (1+\Delta\theta)^{-(k+1)} \\
\Delta|k, \theta &\sim \operatorname{Lomax}\left(\frac{1}{\theta}, k\right)
\end{align*}
\end{linenomath*}

Posterior distribution:
\begin{linenomath*}
\begin{align*}
f(\nu|\Delta) &\propto \frac{1}{\Gamma(k)\theta^{k}} \nu^{k-1} e^{-\nu/\theta} \times \nu e^{-\nu \Delta} \propto \nu^k e^{-\nu(1/\theta+\Delta)} = \nu^{(k+1)-1} e^{-\nu/\left(\frac{\theta}{1+\Delta\theta}\right)} \\
\nu|\Delta &\sim \operatorname{Gamma}\left(k + 1, \frac{\theta}{1 + \Delta \theta}\right)
\end{align*}
\end{linenomath*}

\section{Source code}

Birch is available at\\[0.25pc]
\url{https://birch-lang.org/}    

The source code for the CRBD and BiSSE models is available at\\[0.25pc]
\url{https://github.com/kudlicka/paper-2019-probabilistic}

\begin{figure}[h]
\centering
\includegraphics[width=\textwidth]{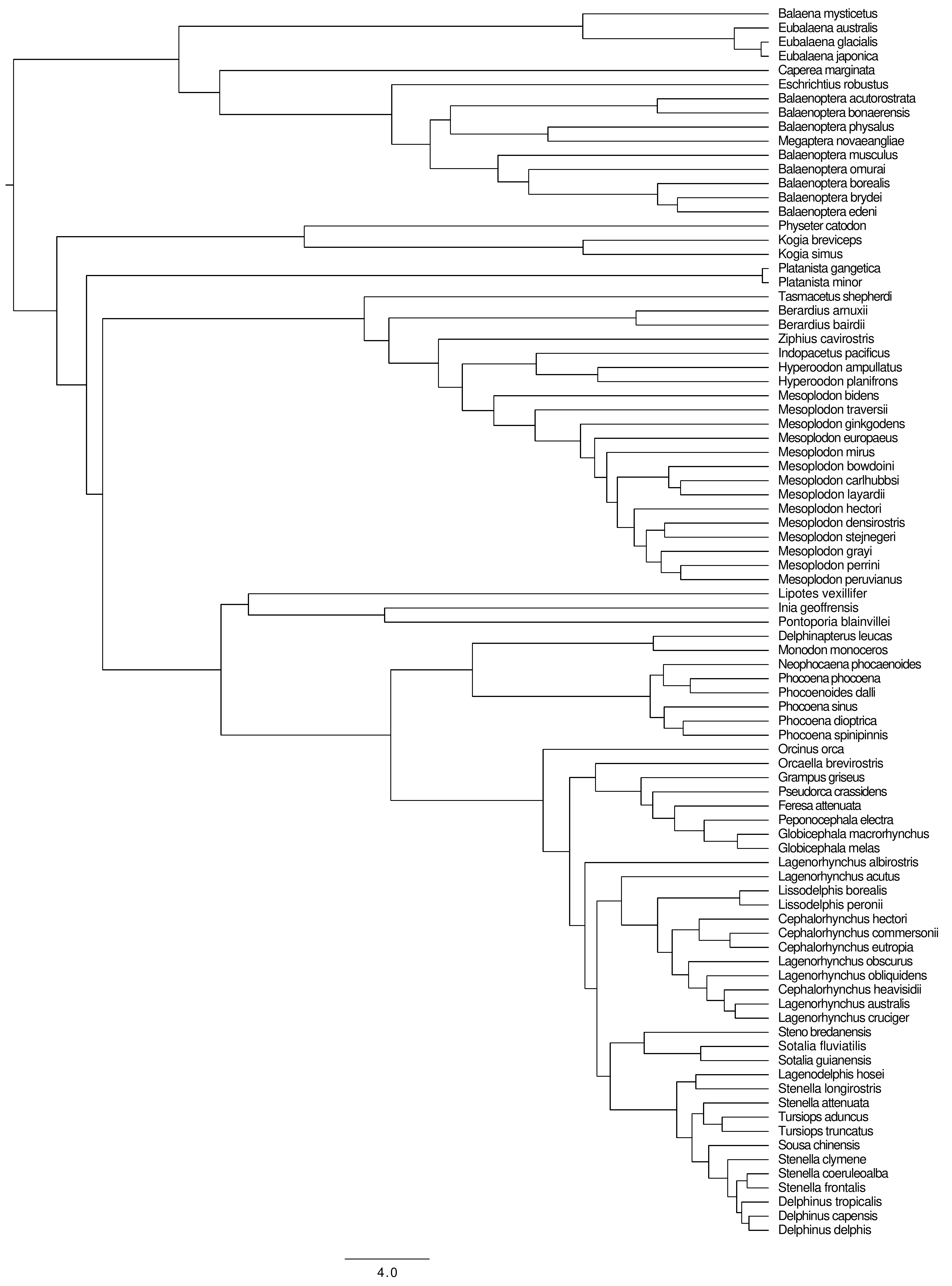}
\caption{Phylogeny of cetaceans (whales, dolphins and porpoises).}
\label{fig:cetaceans}
\end{figure}

\end{document}